\documentclass{article}
\usepackage{graphicx} 
\usepackage{fullpage}
\usepackage{clrscode3e}
\usepackage{hyperref}
\usepackage[x11names]{xcolor}
\usepackage{amssymb, amsmath, amsthm}

\newtheorem{theorem}{Theorem}[section]
\newtheorem{lemma}[theorem]{Lemma}

\usepackage[backend=bibtex,maxbibnames=99,sortcites]{biblatex}

\renewbibmacro*{doi+eprint+url}{
    \printfield{doi}
    \newunit\newblock
    \iftoggle{bbx:eprint}{
        \usebibmacro{eprint}
    }{}
    \newunit\newblock
    \iffieldundef{doi}{
        \usebibmacro{url+urldate}}
        {}
    }
\hypersetup{
    colorlinks,
    linkcolor={red!50!black},
    citecolor={blue!50!black},
    urlcolor={blue!80!black}
}

\title{$\proc{ExpoSort}$: Breaking the quasi-polynomial-time barrier for reluctant sorting}
\author{Mikkel Abrahamsen\thanks{Department of Computer Science, University of Copenhagen, \texttt{miab@di.ku.dk}.
The author is supported by Starting Grant 1054-00032B from the Independent Research Fund Denmark under the Sapere Aude research career programme and is part of Basic Algorithms Research Copenhagen (BARC), supported by the VILLUM Foundation grant 16582.}}
\date{September 2024}

\begin{document}

\maketitle

\begin{abstract}
We introduce the algorithm $\proc{ExpoSort}$, a groundbreaking method that sorts an array of 
$n$ numbers in a spectacularly inefficient $\Theta(2^n)$ time.
$\proc{ExpoSort}$ proudly claims the title of the first reluctant algorithm to decisively surpass the quasi-polynomial running time $\Omega(n^{\log n/(2+\varepsilon)})$ of the notoriously sluggish $\proc{SlowSort}$ algorithm by Broder and Stolfi [ACM SIGACT News, 1984].
In the ongoing quest for the slowest possible sort, $\proc{ExpoSort}$ redefines what it means to take one's time.

Remarkably, $\proc{ExpoSort}$ achieves this feat with one of the simplest pseudocodes among all known sorting algorithms.
However, a slight modification---merely moving one recursive call inside an if statement---transforms $\proc{ExpoSort}$ into an astonishingly well-camouflaged variant of the classic $\proc{InsertionSort}$ with best- and worst-case running times of $\Theta(n)$ and $\Theta(n^3)$, respectively.
This dual nature of $\proc{ExpoSort}$ serves as a reminder of the utmost care required when crafting pessimal algorithms, where a slight lapse in judgment could result in accidentally producing an embarrassingly practical algorithm.
\end{abstract}

\section{The algorithm}

Sorting algorithms are a cornerstone of computer science, and they often serve as a litmus test for understanding algorithmic efficiency.
Among the many algorithms that have been studied and refined, we have seen a range of performance characteristics, from the elegantly efficient to the delightfully cumbersome.
In the latter category, the quest for the slowest conceivable sorting algorithm has intrigued both amateur and professional researchers alike.
It is with this spirit of exploration and jest that we introduce a new contender in the realm of inefficient sorting algorithms:

\begin{codebox}
  \Procname{$\proc{ExpoSort}(A, n)$}
  \li \If $n>1$ \Do
  \li $\proc{ExpoSort}(A,n-1)$
  \li \If $A[n-1]>A[n]$ \Do
  \li swap $A[n-1]$ and $A[n]$ \End
  \li $\proc{ExpoSort}(A,n-1)$
  \End
\end{codebox}

It is easily seen by induction on $n$ that the algorithm sorts an array of 
$n$ numbers.
Furthermore, the running time $T(n)$ clearly satisfies the recurrence $T(n)=2\, T(n-1)+\Theta(1)$.
This solves to 
$T(n)=\Theta(2^n)$, so it is fair to say that $\proc{ExpoSort}$ is a sorting algorithm that embraces inefficiency with enthusiasm.
Its design is a testament to the principle that, in the world of algorithms, sometimes the simplest approach can yield the most excruciatingly inefficient results.
Indeed, we don't know of any sorting algorithm that can be written with a simpler and shorter pseudo-code.

\section{Related work}

The quest for the slowest sorting algorithm is a storied one, with notable contributions from numerous inventive minds.
The first contender might have been an unnamed algorithm, whose worst-case running time of 
$\Theta(n^3)$ was, quite intentionally, designed to be amusingly sluggish.
The algorithm is described in a paper by Bentley~\cite{10.1145/358027.381121}, who attributed it to Steele, Finkel, Woods, Crispin, Stallman, and Goodfellow~\cite{hacker}.
Another algorithm admired for its calm nature is $\proc{StoogeSort}$, with its peculiar running time $\Theta(n^{\log_{3/2} 3})=\Theta(n^{2.7095\ldots})$.
The origin of $\proc{StoogeSort}$ seems unclear, but it is for instance described in~\cite[Problem 7-4]{DBLP:books/daglib/0023376}.

In their seminal 1984 paper, Broder and Stolfi~\cite{DBLP:journals/sigact/BroderS84} formalized the concept of \emph{reluctant algorithms}, a breed of algorithms distinguished not by their efficiency but by their remarkable ability to delay the inevitable.
They defined a reluctant algorithm as one that ``wastes time in a way that is sufficiently contrived to fool a naive observer.''
According to their criteria, simply adding extraneous loops to achieve a galactic running time was considered too blatant: ``any fool can see that the algorithm is just wasting time.''

Instead, they advocated for algorithms that, despite their profound aversion to rapid completion, still make steady progress towards their goal.
The $\proc{SlowSort}$ algorithm, with its quasi-polynomial time complexity between $\Omega(n^{\log n/(2+\varepsilon)})$ for any $\varepsilon>0$ and $O(n^{\log n/2})$, was one of their prime examples of this reluctant approach.
$\proc{SlowSort}$ follows the so-called \emph{multiply-and-surrender} paradigm introduced in the same paper.
The core idea of the paradigm is to break down the original problem into two or more subproblems, each only slightly simpler than the original.
This process is recursively repeated, generating more and more subproblems.
In the end we have an astonishing number of subproblems that are so trivial that there is no choice but to report their solutions, at which point we must finally surrender.

Our new algorithm $\proc{ExpoSort}$ demonstrates that this paradigm can get us even further.
The algorithm not only aims to outshine $\proc{SlowSort}$ in its quest for slowness but also adheres to the philosophy behind reluctant algorithms by maintaining a steady yet painfully slow progress.
Unlike the notorious $\proc{BogoSort}$~\cite{DBLP:conf/fun/GruberHR07}, which exacerbates its inefficiency by repeatedly permuting the array randomly until sorted---resulting in an expected running time of 
$\Theta(n\cdot n!)$---$\proc{ExpoSort}$ advances with a measured pace, never increasing the number of inversions in the array.

\section{The devil in the details}

Interestingly, the following slight modification of $\proc{ExpoSort}$ drastically changes its properties:

\begin{codebox}
  \Procname{$\proc{CubeSort}(A, n)$}
  \li \If $n>1$ \Do
  \li $\proc{CubeSort}(A,n-1)$
  \li \If $A[n-1]>A[n]$ \Do
  \li swap $A[n-1]$ and $A[n]$  \li $\proc{CubeSort}(A,n-1)$
  \End
  \End
\end{codebox}

As the following lemma states, simply increasing the indentation of the recursive call at line 5 of $\proc{ExpoSort}$ deteriorates the running time exponentially, resulting in the algorithm $\proc{CubeSort}$ which is dangerously close to being practical.
When the input array is sorted, $\proc{CubeSort}$ is even asymptotically optimal!
We hope this example serves as both a warning and a source of inspiration for those who delve into the delicate art of constructing pessimal algorithms.

\begin{lemma}
The best-case running time of $\proc{CubeSort}$ is $\Theta(n)$, and the worst-case is $\Theta(n^3)$.
\end{lemma}

\begin{proof}
Let $T_1(n)$ be the running time when $A$ is already sorted.
The recursive call at line 2 takes $T_1(n-1)$ time and there is no call at line 5, so we get the recurrence $T_1(n)=T_1(n-1)+\Theta(1)$, solving to $T_1(n)=\Theta(n)$.
This is also the best-case running time.

Let $T_2(n)$ be the running times when $A[1:n-1]$ is sorted but the minimum element is $A[n]$.
In this case, the first recursive call at line 2 takes $T_1(n-1)$ time, then $A[n-1]$ is swapped with $A[n]$, and then the recursive call at line 5 takes $T_2(n-1)$ time.
We thus have $T_2(n)=T_1(n-1)+T_2(n-1)+\Theta(1)=T_2(n-1)+\Theta(n)$, which solves to $T_2(n)=\Theta(n^2)$.

Finally, let $T_3(n)$ be the running time when $A$ is reversely sorted.
Here the recursive call at line 2 takes $T_3(n-1)$ time, after which $A[1:n-1]$ is sorted, and then $A[n-1]$ is swapped with $A[n]$, so the recursive call at line 5 takes $T_2(n-1)$ time.
We thus have $T_3(n)=T_3(n-1)+T_2(n-1)+\Theta(1)=T_3(n-1)+\Theta(n^2)$, which solves to $T_3(n)=\Theta(n^3)$.
Clearly, this is also an upper bound for an arbitrary input array.
\end{proof}

We will now see that $\proc{ExpoSort}$ and $\proc{CubeSort}$ can both be seen as leisurely versions of the classical (and hectic!) $\proc{InsertionSort}$.
Let us first refresh how $\proc{InsertionSort}$ works.

\begin{codebox}
\Procname{$\proc{InsertionSort}(A, n)$}
  \li \For $j = 2$ \To $n$ \Do
  \li $i = j$
  \li \While $i > 1$ \textbf{and} $A[i-1] > A[i]$ \Do
  \li swap $A[i-1]$ and $A[i]$
  \li $i = i - 1$
  \End
  \End
\end{codebox}

We can now state the similarity between the algorithms.

\begin{lemma}
$\proc{ExpoSort}$ and $\proc{CubeSort}$ both make the same swap sequence as $\proc{InsertionSort}$.
\end{lemma}

\begin{proof}
We first observe that $\proc{ExpoSort}$ and $\proc{CubeSort}$ make the same swap sequences by induction on $n$.
The base case $n=1$ is trivial, so suppose inductively that the claim holds for arrays of size at most $n-1$ and consider what happens when the algorithms are applied to an array $A$ of size $n$.
The algorithms both make a recursive call on the array $A[1:n-1]$ at line 2, which, by the induction hypothesis, make the same swap sequence.
The algorithms either both swap or don't swap $A[n-1]$ and $A[n]$ at line 4.
If they swap, they will both make the recursive call at line 5, which by the induction hypothesis produce the same swap sequence.
If they don't, the array is already sorted, and then the recursive call at line 5 of $\proc{ExpoSort}$ will not make any swaps.
In either case, the resulting swap sequences are the same.

We now observe that $\proc{CubeSort}$ makes the same swap sequence as $\proc{InsertionSort}$.
We again proceed by induction on $n$.
The base case $n=1$ is trivial.
Suppose that the claim is true for arrays of size $n-1$ and consider how $\proc{CubeSort}$ sorts an array $A$ of size $n$.
Using the induction hypothesis, $\proc{CubeSort}$ starts by sorting the first $n-1$ numbers with the same swaps as when $\proc{InsertionSort}$ sorts an array of $n-1$ numbers.
This is also how $\proc{InsertionSort}$ starts sorting $n$ numbers.
We thus need to verify that $\proc{CubeSort}$ proceeds by swapping the last element $x=A[n]$ with the elements to the left of it, until it reaches its designated position.
This is clear since $\proc{CubeSort}$ and $\proc{InsertionSort}$ both only swap neighbouring elements, and then there is a unique way to swap $x$ to the right place (without increasing the number of inversions, which none of the algorithms do).
\end{proof}

\subsection*{Acknowledgements}

I thank Thomas Dybdahl Ahle and Viktor Fredslund-Hansen for their valuable discussions on the serious matters presented in this paper.
I am also grateful to my students, who inspired me to explore strange algorithms for use in teaching and exams.

\printbibliography

\end{document}